\documentclass[12pt]{article}
\usepackage{longtable,supertabular}
\usepackage{amsmath}
\usepackage{amssymb}
\usepackage{latexsym}
\usepackage{a4wide,color}
\usepackage{multirow}
\usepackage{ntheorem}
\usepackage[square,numbers,sort]{natbib}
\usepackage[title]{appendix}
\usepackage{array}
\usepackage{makecell}

\usepackage{booktabs}

\usepackage{xcolor}
\usepackage[colorlinks=true,linkcolor=blue,citecolor=blue,urlcolor=blue]{hyperref}

\usepackage{float}

\newtheorem{theorem}{Theorem}[section]

\newtheorem{lemma}{Lemma}[section]
\newtheorem{definition}{Definition}[section]
\newtheorem{conjecture}[theorem]{Conjecture}
\newtheorem*{proof}{Proof.}

\newtheorem{example}{Example}[section]

\newtheorem{remark}{Remark}[section]

\newcommand{\F}{\mathbb{F}}
\newcommand{\M}{\mathbb{M}}

\newcommand{\ord}{\mathrm{ord}}

\newcommand{\bC}{{\mathbf{{C}}}}
\newcommand{\bF}{{\mathbf{{F}}}}

\newcommand{\ba}{{\mathbf{a}}}

\newcommand{\bc}{{\mathbf{c}}}

\newcommand{\bV}{{\mathbf{V}}}
\newcommand{\bM}{{\mathbf{M}}}
\newcommand{\bD}{{\mathbf{D}}}

\newcommand{\lcm}{{\mathrm{lcm}}}

\newcommand{\diag}{{\mathrm{diag}}}
\newcommand{\charr}{{\mathrm{char}}}
\newcommand{\Tr}{{\mathrm{Tr}}}
\newcommand{\N}{{\mathrm{N}}}

\newcommand{\bone}{{\mathbf{1}}}
\newcommand{\qed}{\hfill $\square$}

\title{\bf On the Minimum Distances of Some Families of BCH Codes}
\author{
	Yaqi Chen,  Hao Chen, Cunsheng Ding and Huimin Lao\thanks{Yaqi Chen is with the College of Cyber Security, Jinan University, Guangzhou, Guangdong Province, 510632, China.
		Hao Chen is with the College of Information Science and
		Technology, Jinan University, Guangzhou, Guangdong Province, 510632, China.
		Cunsheng Ding is with the Department of Computer Science and Engineering, The Hong Kong University of Science and Technology, Hong Kong, China.
		Huimin Lao is with Strategic Centre for Research in Privacy-Preserving Technologies and Systems, Nanyang Technological University, Singapore.
		(e-mail: chenyq@stu.jnu.edu.cn, haochen@jnu.edu.cn, cding@ust.hk,  huimin.lao@ntu.edu.sg).  The research of Hao Chen was supported by NSFC Grant 62032009. The research of C. Ding was supported by 
		the Hong Kong Research Grants Council under Grant No. 16301123.}
}

\begin{document}
	
	\maketitle
	\begin{abstract}
        
BCH codes form an important class of cyclic codes, which have applications in communication and data storage 
systems. Although the BCH bound provides a lower bound on the minimum distance of BCH codes, determining the true minimum distances of BCH codes is a very challenging problem. In this paper, we settle the minimum 
distances of a number  of infinite families of narrow-sense BCH codes. \\

        By explicitly constructing the locator polynomials for minimum weight codewords, we obtain many families of primitive and non-primitive BCH codes with $d=\delta$, where $d$ is the minimum distance of a $q$-ary BCH code of length $n$, designed distance $\delta$, and offset $b$, denoted by $\mathbf{C}_{(q, n, \delta, b)}$. 
        For primitive BCH codes, we obtain infinite families of BCH codes over $\F_3$ and $\F_4$ satisfying $d=\delta$, where $\delta \in \{5,6,7,8\}$.
        Moreover, we construct several infinite families of $q$-ary BCH codes with $d=\delta$, where $2 \le \delta \le q-1$.
        For $\delta=q^t+1$, we prove that the BCH code $\bC_{(q, q^m-1, q^t+1, 1)}$ has $d=\delta$ for all $m$ satisfying $m \equiv 0 \pmod{pt}$, where $p$ denotes the characteristic of $\F_q$.
        In the paper by Ding et al., IEEE Trans. Inf. Theory 61(5): 2351-2356, it was conjectured that the minimum distance of $\bC_{(q, q^m-1, q^t+1, 1)}$ is always equal to its Bose distance $d_B$.
        Our result confirms this conjecture for the case $m \equiv 0 \pmod{pt}$.
        For non-primitive BCH codes, we construct a family of BCH codes   $\bC_{(q,\frac{q^p-1}{\lambda},p+1,1)}$ with $d=\delta=p+1$, where $p$ is an odd prime, $q=p^e$ with $p \nmid e$ and $\lambda \mid q-1$. \\

	\end{abstract}
	
	\newpage

\section{Introduction}\label{s1}

\subsection{Introduction of basic concepts}

Throughout this paper, let $\F_q$ denote the finite field with $q$ elements, where $q$ is a power of a prime, and let $\F_q^*$ denote the multiplicative group of $\F_q$. The Hamming weight $wt(\ba)$ of a vector $\ba = (a_0, \ldots, a_{n-1}) \in \F_q^n$ is the cardinality of its support,
$$\mathrm{supp}({\bf a}):=\{i: a_i \neq 0\}.$$
The Hamming distance $d({\bf a}, {\bf b})$ between two vectors ${\bf a}$ and ${\bf b}$ is $d({\bf a}, {\bf b})=wt({\bf a}-{\bf b})$. 
Any subset ${\bf C} \subseteq {\bf F}_q^n$ is called a $q$-ary code of length $n$.  
The minimum distance of a code ${\bf C} \subseteq {\bf F}_q^n$ is defined as 
$$d({\bf C}):=\min_{{\bf a} \neq {\bf b}} \{d({\bf a}, {\bf b}),  {\bf a} \in {\bf C}, {\bf b} \in {\bf C} \}.$$
A $q$-ary linear code $\bC$ with parameters  $[n, k, d]_q$ is a $k$-dimensional linear subspace of $\mathbb{F}_q^n $ with minimum distance $d$. The dual code of $\bC$ is defined as
\[ \bC^\perp := \{\mathbf{b} \in \mathbb{F}_q^n\,|\,\mathbf{b} \cdot \mathbf{c}^\mathrm{T}=\mathbf{0},\,\forall \mathbf{c} \in \bC\}. \]

	Let $\bC$ be an $[n, k, d]_q$ linear code.
	We say that $\bC$ is distance-optimal if no $[n, k, \geq d+1]_q$ linear code exists.
	If there is an  $[n, k, d+1]_q$ linear code and there is no $[n, k, \geq d+2]_q$ linear code, then $\bC$ is said 
	to be almost distance-optimal.
	Similarly, if there is an $[n, k, d+2]_q$ linear code, but there is no $[n, k,  \geq d+3]_q$ linear code, then $\bC$ is called a near distance-optimal code.

    Cyclic codes, first introduced by Prange in \cite{Prange}, are a significant class of linear error-correcting codes due to their structural properties and efficient decoding algorithms. An $[n,k,d]_q$ linear code $\bC$ is said to be cyclic if $(c_0, c_1, \ldots, c_{n-1}) \in \bC $ implies $(c_{n-1}, c_0,\ldots, c_{n-2}) \in \bC$.
    Each vector $(c_0, c_1, \ldots, c_{n-1}) \in \F_q^n$ can be identified with the polynomial $$c_0 + c_1x + \cdots + c_{n-1}x^{n-1} \in \F_q[x]/(x^n-1).$$ 
With this identification, any cyclic code ${\bf C} \subseteq {\bf F}_q^n$ is an ideal of the ring $\F_q[x]/(x^n-1)$.    
    Note that every cyclic code is a principal ideal in the ring ${\F}_q[x]/(x^n-1)$, generated by a monic polynomial $g(x)$ of the smallest degree,  which divides $x^n-1$.  This polynomial $g(x)$ is called the generator polynomial and $h(x)=(x^n-1)/g(x)$ is called the parity-check polynomial of the cyclic code $\bC$.

    Let $n>1$ be an integer such that $\gcd(n,q)=1$ and let $m=\ord_n(q)$ denote the smallest integer satisfying $q^m \equiv 1 \pmod{n}$. Let $\alpha$ be a primitive element of $\F_{q^m}$, then $\beta=\alpha^{\frac{q^m-1}{n}}$ is a primitive $n$-th root of unity. This leads to the factorization $x^n-1=\prod_{i=0}^{n-1}(x-\beta^i)$ over $\F_{q^m}$. Let $\M_i(x)$ denote the minimal polynomial of $\beta^i$ over $\F_q$ for $0 \le i \le n-1$, which depends on the choice of the primitive element $\alpha$ of 
    $\bF_{q^m}$.

    BCH (Bose-Chaudhuri-Hocquenghem) codes are subclasses of cyclic codes (\cite{BC1,BC2,Hoc}).
    A cyclic code of length $n$ over $\F_q$ is called a BCH code with designed distance $\delta$ if its generator polynomial takes the form
    $$g(x)=\lcm(\M_b(x),\M_{b+1}(x),\dots,\M_{b+\delta-2}(x))$$
    for some integers $b$ and $2 \le  \delta \le n$, where $\lcm$ denotes the least common multiple of polynomials.  
    Throughout this paper, we denote by $\bC_{(q,n,\delta,b)}$ the BCH code over $\F_q $ with length $n$, designed distance $\delta$ and offset $b$. When $b = 1$, the code $\bC_{(q,n,\delta,b)}$ is referred to as a narrow-sense BCH code. If $n = q^m - 1$, $\bC_{(q,n,\delta,b)}$ is called a primitive BCH code.

By definition, the BCH code $\bC_{(q,n,\delta,b)}$ depends on the choice of $\beta$ and $\alpha$. However, the parameters 
(i.e., the length, the dimension and the minimum distance) of $\bC_{(q,n,\delta,b)}$ are independent of the choice of the $n$-th root of unity $\beta$.  The defining set of $\bC_{(q,n,\delta,b)}$ with respect to $\beta$ is defined by
    $$T=\{0 \le i \le n-1: g(\beta^i)=0 \},$$ 
which also depends on the choice of $\alpha$ and $\beta$. 
    
\subsection{The motivations of this paper}    

BCH codes have been a hot research topic in the past 70 years due to the following:
\begin{itemize}
\item BCH codes are closely related to a number of areas of mathematics such as algebra, algebraic number theory, 
algebraic function fields, algebraic geometry, combinatorics,  elementary number theory, finite fields,  graph theory and group theory.  
\item Some BCH codes are widely used in communication and data storage systems.  
\end{itemize} 
Although a lot of progress in the study of BCH codes has been made in the past 70 years \cite{Ding1},  the minimum distance of only a small number of families of BCH codes is settled.   

Throughout this paper,  let $i \bmod{n}$ denote the unique integer $r$ with $0 \leq r \leq n-1$ such that 
$i-r \equiv 0 \pmod{n}$. 
    If there exist integers $b$ and $\delta$ with $2 \le \delta \le n$ satisfying $\{(i+b) \bmod{n} : 0 \le i \le \delta - 2\} \subseteq  T$, where $T$ is the defining set of  $\bC_{(q,n,\delta,b)}$, then the minimum distance of $\bC_{(q,n,\delta,b)}$ is at least $\delta$. This is known as the BCH bound \cite{BC1}.  This BCH bound provides a lower bound on the minimum distance,  which could be tight for certain types of BCH codes and 
    could be faraway from the true minimum distance for some types of BCH codes.      
    It is known that $\bC_{(q,n,\delta_1,b)}=\bC_{(q,n,\delta_2,b)}$ may hold for some $\delta_1$ and 
    $\delta_2$ with $\delta_1 \ne \delta_2$. The largest designed distance defining a BCH code is called the Bose distance $d_B$ \cite[p. 171]{HP}. Thus we have $d \ge d_B \ge \delta$, which provides an approximation to the true minimum distance. 
    
    Determining the true minimum distances of BCH codes is a notoriously hard problem.  Below is a list of results on the true minimum distances of some BCH codes.

    \begin{itemize}
    	\item The code $\bC_{(2,2^m-1,\delta,1)}$ with $\delta=2t+1$ has minimum distance $d=\delta$ if either $\sum_{i=0}^{t+1} \binom{2^m-1}{i} > 2^{mt}$ or $m > \log_2(t+1)! + 1$ \cite[p. 259]{MScode}.
    	
    	\item The code $\bC_{(q,q^m-1,\delta,1)}$ has $d=\delta$ if $\delta=q^l-1$ for $1 \le l \le m-1$ 
    	(\cite{Peterson} and \cite[p. 260]{MScode}).
    	
    	\item The code $\bC_{(2,2^m-1,\delta,1)}$ has $d=\delta$ if $\delta=2^{m-1-s}-2^{m-1-s-i}-1$ for $1 \le i \le m-s-2$ and $1 \le s \le m-2i$ \cite{KL}.   
    	
   	\item The code  $\bC_{(q,q^m-1,\delta,1)}$ with $\delta=q^m-q^{m-1}-q^i-1$  has $d=\delta$ if $i$ is a nonnegative integer satisfying $\frac{m-2}{2} \le i \le m - \left \lfloor \frac{m}{3} \right \rfloor -1 $ \cite{LS}. This result has generalized Berlekamp's earlier work in \cite{Berlekamp}.

    	    	\item If a BCH code $\bC_{(q,q^m-1,\delta,1)}$  has $d=\delta$ and $\delta+1$ is divisible by $p$, where $q$ is a power of a prime $p$,  then a code $\bC_{(q,q^m-1,\delta',1)}$
    	    	with designed distance $\delta'=(\delta+1)q^{m-h}-1$ and $h \ge \delta$ also has minimum distance $\delta'$ [Note that the condition that $\delta +1$ be divisible by $p$ is always satisfied by all binary codes $\bC_{(2,2^m-1,\delta,1)}$] \cite{Peterson}.

    	\item The code $\bC_{(2,2^m-1,\delta,b)}$ has $d=\delta$ if $\delta$ divides $2^m-1$ and $b \geq 0$  
    	 \cite[Thm.3, Chapter 9]{MScode}.

    	\item The code $\bC_{(q,n,\delta,b)}$ has $d=\delta$ if $\delta$ divides $\gcd(n,b-1)$ \cite{LiLiDing17cc}.   	 
    	
    	\item Any narrow-sense BCH code of length $n$ over $\F_q$ has $d=\delta$ whenever $\delta$ divides $n$ \cite[p. 247]{Betten}.

    \end{itemize}

A necessary and sufficient condition for $d\in \{3,4\}$ for the code $\bC_{(q,q+1,3,h)}$ and for the code 
$\bC_{(q, q^m+1, 3, h)}$ were presented in  \cite{Cao1} and \cite{Cao2}, respectively.  The minimum distance of some other families of BCH codes has also been settled \cite{Ding1,LS}.

In \cite{Ding0},  the authors determined the Bose distances of a class of narrow-sense primitive BCH codes $\bC_{(q, q^m-1, q^t+1,1)}$ and presented the following conjecture.
	
	\begin{conjecture}\label{c-0}
		The minimum distance $d$ of the BCH code $\bC_{(q, q^m-1, q^t+1,1)}$ is always equal to its Bose distance $d_B$, which is given by, $$d_B = \left\lfloor \frac{q^m - 2}{q^{m-t} - 1} \right\rfloor + 1.$$
	\end{conjecture}

   The minimum distance $d$ has been proven to equal the Bose distance $d_B$ only in some specific cases \cite[Thm. 13]{Ding0}.  When $t \equiv 0 \pmod{m-t}$, $d=d_B=\frac{q^m - 1}{q^{m-t} - 1}$. When $m \equiv 0 \pmod{2t}$, $d=d_B=\delta=q^t+1$.
   In both cases, we have $d_B \mid n$.   
   Recently, another special case of Conjecture \ref{c-0} was proved by Shany and Berman \cite{S2025}, who confirmed the conjecture for the case that $q=2$, $m \ge 4$ and $t=m-2$. 
   However,  the true minimum distance of this code is still open in other subcases of the case $\delta \nmid n$.   
   
In short, the motivations of this paper are the following:
\begin{itemize}
\item BCH codes are very important in theory and applications. 
\item The true minimum distance of only a small number of families of BCH codes is known. 
\item Conjecture \ref{c-0} is still open in certain cases. 
\end{itemize}   
	
\subsection{The contributions of this paper}\label{sec-contribu}

    In this paper, we settle the minimum distances of some families of BCH codes 
    by explicitly determining the locator polynomial of a minimum weight codeword.  
    Specifically, we prove that  the minimum distances of the following infinite families of BCH codes are equal to 
    their designed distances.

	\begin{itemize}
		\item 
		The ternary primitive BCH codes $\bC_{(3, 3^m-1, \delta, 1)}$ with $\delta=5,7,8$ and the quaternary primitive BCH codes $\bC_{(4, 4^m-1, \delta, 1)}$ with $\delta=5,6,7$,  see Theorem \ref{T-ternary} and Theorem \ref{T-quaternary}.

		\item
		The $q$-ary primitive BCH codes $\bC_{(q, q^m-1, \delta, 1)}$ with  
		$2 \le \delta \le q-1$,  see Theorem \ref{T-delta<q}.
		
		\item
	The	$q$-ary primitive BCH codes  $\bC_{(q, q^m-1, \delta, 1)}$ with $\delta=q^t+1$,  see Theorem \ref{T-q^t+1}.  The study of this family of codes confirms  Conjecture \ref{c-0} in the case $m \equiv 0 \pmod{pt}$, where $p$ is the characteristic of $\F_q$.
		
		\item
	The $q$-ary non-primitive BCH codes $\bC_{(q, n, \delta, 1)}$ with $\delta=p+1$  and length $n=\frac{q^p-1}{\lambda}$,
		where $p$ is an odd prime, $q=p^e$ with $p \nmid e$ and $\lambda \mid q-1$, see Theorem \ref{T-p+1/lambda}.
	\end{itemize}

The reader is informed that the minimum distance of the codes  above is known in the case  
that $\delta$ is a factor of the length $n$ \cite{Betten} (we call it the trivial case).  
In the following table, we list these trivial cases and our cases.  The results in our cases are new.

\begin{longtable}{|l|l|l|l|}
		\caption{\label{tab:s1} The BCH codes with $d=\delta$ studied in this paper} \\ \hline
		Theorem & The codes & Our cases & Trivial cases ($\delta \mid n$) \\ \hline
		\endfirsthead
		
		\hline
		\endlastfoot
		
		% --- Theorem 4.1 ---
		\multirow{3}{*}{Thm \ref{T-ternary}} & \multirow{3}{*}{$\bC_{(3, 3^m-1, \delta, 1)}$}
		&  $2 \mid m$ or  $3 \mid m$ for $\delta=5$ & $4 \mid m$ for $\delta=5$\\ \cline{3-4}
		& & $3 \mid m$ or $4 \mid m$ for $\delta=7$ & $6 \mid m$ for $\delta=7$\\ \cline{3-4}
		& & $3 \mid m$ for $\delta=8$ & $2 \mid m$ for $\delta=8$\\ \hline
		
		% --- Theorem 4.2 ---
		\multirow{3}{*}{Thm \ref{T-quaternary}} & \multirow{3}{*}{$\bC_{(4, 4^m-1, \delta, 1)}$}
		& $2 \mid m$ or $3 \mid m$ for $\delta=5$ & $2 \mid m$ for $\delta=5$ \\ \cline{3-4}
		& & $2 \mid m$ or $3 \mid m$ for $\delta=6$ & Never holds \\ \cline{3-4}
		& & $2 \mid m$ or $3 \mid m$ for $\delta=7$ & $3 \mid m$ for $\delta=7$ \\ \hline
		
		% --- Theorem 4.3 ---
		Thm \ref{T-delta<q} & $\bC_{(q, q^m-1, \delta, 1)}$ & $m \ge 1$,  $2 \le \delta \le q-1$ & $\delta \mid q-1$ \\ \hline

		% --- Theorem 4.5 ---
		Thm \ref{T-q^t+1} & $\bC_{(q, q^{m}-1, q^t+1, 1)}$ &
		\begin{tabular}[c]{@{}l@{}} $p=\charr(\F_q)$, \\$t \ge 1$, $m/pt$ is odd \end{tabular} & \begin{tabular}[c]{@{}l@{}}  $p=\charr(\F_q)$, \\$m/pt$ is even \end{tabular}\\ \hline
		
		% --- Theorem 5.1 ---
		Thm \ref{T-p+1/lambda} & $\bC_{(q, \frac{q^p-1}{\lambda}, p+1, 1)}$ &
		\begin{tabular}[c]{@{}l@{}} $p$ is an odd prime, $e \ge 1$, \\ $q=p^e$, $p \nmid e$,  $\lambda \mid q-1$ \end{tabular} &
		\begin{tabular}[c]{@{}l@{}} $p$ is an odd prime, \\$2 \mid e$,  $\lambda \mid \frac{q-1}{p+1}$ \end{tabular}   \\ \hline
		
	\end{longtable}

\subsection{The organisation of this paper}

    The rest of the paper is organized as follows.
    Section \ref{s2} recalls the necessary preliminaries on cyclotomic cosets and locator polynomials.
    Section \ref{s3} introduces some auxiliary results including the inverse of the modified Vandermonde matrix.
    Section \ref{s4} presents some families of primitive BCH codes with $d=\delta$.
    Section \ref{s5} proves that a family of non-primitive BCH codes has the property that $d=\delta$.
    Section \ref{s6} concludes the paper.

\section{Preliminaries}\label{s2}
\subsection{Cyclotomic cosets}

    Let $n$ be a positive integer and ${\mathbb{Z}}_n={\mathbb{Z}}/n{\bf \mathbb{Z}}=\{0, 1, \ldots, n-1\}$ denote the set of residue classes modulo $n$. For any $i \in  \mathbb{Z}_n$, the $q$-cyclotomic coset of $i$ modulo $n$ is defined by
    $$C_i=\{ iq^j \bmod n : 0 \le j \le l_i-1\},$$
    where $l_i$ is the smallest positive integer such that $iq^{l_i} \equiv i \pmod n$. The smallest integer in $C_i$ is called the coset leader of $C_i$. The set of all $q$-cyclotomic cosets modulo $n$ forms a partition of ${\mathbb{Z}}_n$.

    Let $m=\ord_n(q)$ and let $\beta \in  \F_{q^m}$ be an $n$-th primitive root of unity. The minimal polynomial $\M_i(x)$ of $\beta^i$ over $\F_q$ is irreducible over $\F_q$ and is given by
    $$\M_i(x)=\prod_{j \in C_i}(x-\beta^j) \in \F_q[x].$$ 
    For a BCH code $$\bC_{(q,n,\delta,b)}=\left \langle g(x) \right \rangle \subseteq {\F_q[x]/(x^n-1)},$$
    the generator polynomial $g(x)$ is the product of some of these minimal polynomials $\M_i(x)$. Consequently, the degree of $g(x)$ and the dimension of $\bC_{(q,n,\delta,b)}$ can be determined by the sizes  of the cyclotomic cosets associated with $g(x)$.

    The following two lemmas concern the sizes of cyclotomic cosets.

    \begin{lemma}\cite{HP}
        The size of each $q$-cyclotomic coset $C_i$ modulo $n$ is a divisor of $\ord_n(q)$. In particular, $|C_1|=\ord_n(q)$.
    \end{lemma}

    \begin{lemma}\label{L-2-size}\cite{AA2007}
        Suppose $q^{\lceil m/2 \rceil} < n \le q^m - 1$, where $m=\ord_n(q)$.
        Then for all $1 \le i \le n q^{\lceil m/2 \rceil}/(q^m - 1)$, $|C_i|=m$. 
        In addition, every $i$ with
$i \not\equiv 0 \pmod{q}$ in this range is a coset leader.
    \end{lemma}
	
	As a direct consequence of Lemma \ref{L-2-size}, the dimension of the BCH code  $\bC_{(q,n,\delta,1)}$ can be obtained when $\delta$ is small compared to $n$. Specifically, we have the following result.

	\begin{lemma}\label{L-dim}\cite{AA2007}
		Suppose $q^{\lceil m/2 \rceil} < n \le q^m - 1$, where $m=\ord_n(q)$. The narrow-sense BCH code 
		$\bC_{(q,n,\delta,1)}$ with $\delta$ in the range $2 \le \delta \le \min \{\lfloor  n q^{\lceil m/2 \rceil}/(q^m - 1) \rfloor , n  \}$ has dimension
		$$
		k=n-m\lceil (\delta-1)(1-1/q)  \rceil.
		$$
	\end{lemma}

\subsection{Locator polynomials of vectors in $\F_q^n$}

Locator polynomials are a fundamental tool in the study of the minimum distance of BCH codes \cite{ACS92,MScode}. To simplify the differential analysis in our constructive proof, we adopt the following definition, which differs slightly from most references.

\begin{definition}
    Let $\bc=(c_0,c_1,\dots,c_{n-1})\in \F_q^n$ be a vector of Hamming weight $w$ with nonzero components at positions $i_1,i_2,\dots,i_w$.
    Let $m=\mathrm{ord}_n(q)$ and let $\beta$ be a primitive $n$-th root of unity in $\mathbb{F}_{q^m}$.
    Then
    $$x_1=\beta^{i_1}, \dots, x_w=\beta^{i_w},$$
    are called the locators of $\bc$.
    The locator polynomial of $\bc$ is defined by
    $$H(x):=\prod_{j=1}^w(x-x_j)=\prod_{j=1}^w(x-\beta^{i_j}).$$

\end{definition}

In some references, e.g.  \cite{AS1994}, the term locator polynomial refers to the reciprocal polynomial $x^{\deg(H(x))}H(x^{-1})$, whose roots are $1/\beta^{i_j}$.

\section{Auxiliary results}\label{s3}

In this section, we introduce several theoretical tools used to determine the minimum distance of certain BCH codes. We first derive an explicit inverse of a modified Vandermonde matrix. This formula leads to a sufficient condition for the BCH code $\bC_{(q,n,\delta,1)}$ to have minimum distance $d=\delta$.

\begin{lemma}\label{L-2.1}\cite{Rawashdeh2019}
    Let $\bM$ be an $r \times r$ Vandermonde matrix defined as
    $$
    \bM = \begin{pmatrix}
    1 & 1 & \cdots & 1 \\
    x_1 & x_2 & \cdots & x_r \\
    \vdots & \vdots & \ddots & \vdots \\
    x_1^{r-1} & x_2^{r-1} & \cdots & x_r^{r-1}
    \end{pmatrix},
    $$
    where $x_1, \ldots, x_r \in \mathbb{F}_{q^m}$ are pairwise distinct nonzero elements. Then $\bM$ is invertible, and the $(i,j)$-th entry of its inverse is given by
    \begin{equation*}
        (\bM^{-1})_{ij} = (-1)^{r-j} \frac{U_{r-j, i}}{\prod_{k \neq i} (x_i - x_k)}, \quad 1 \le i,j \le r,
    \end{equation*}
    where $U_{k, i}$ denotes the $k$-th elementary symmetric polynomial in the variables $\{x_1, \dots, x_r\} \setminus \{x_i\}$. Specifically, $U_{0,i} = 1$, and for $k > 0$,
    $$
    U_{k, i} = \sum_{\substack{1 \leq j_1 < \dots < j_k \leq r \\ j_1, \dots, j_k \neq i}} x_{j_1} x_{j_2} \cdots x_{j_k}.
    $$
\end{lemma}

Building on this result, we now consider the following modified Vandermonde matrix.

\begin{lemma}\label{L-2.2}
    Let $\bM=(x_{j}^{i-1})_{1 \le i, j \le r}$ be the Vandermonde matrix defined in Lemma \ref{L-2.1}, and let $\bD = \diag(x_1, x_2, \dots, x_r)$ be a diagonal matrix. Define the modified Vandermonde matrix $\bV$ as $\bV := \bM \cdot \bD$, that is,
    $$
    \bV = \begin{pmatrix}
    x_1 & x_2 & \cdots & x_r \\
    x_1^2 & x_2^2 & \cdots & x_r^2 \\
    \vdots & \vdots & \ddots & \vdots \\
    x_1^r & x_2^r & \cdots & x_r^r
    \end{pmatrix}.
    $$
    Then $\bV$ is invertible, and the $(i,j)$-th entry of $\bV^{-1}$ is given by
    \begin{equation*}
        (\bV^{-1})_{ij} = x_i^{-1} (\bM^{-1})_{ij} = (-1)^{r-j} \frac{x_i^{-1} U_{r-j, i}}{\prod_{{k=1 , k \neq i}}^{r} (x_i - x_k)}, \quad 1 \le i,j \le r,
    \end{equation*}
    where $U_{r-j, i}$ was defined in Lemma \ref{L-2.1}.
    Furthermore, the sum of the entries in the $i$-th row of $\bV^{-1}$ satisfies
$$
\sum_{j=1}^{r} (\bV^{-1})_{ij} = \frac{\prod_{{k=1, k \neq i}}^{r} (1 - x_k)}{x_i \prod_{{k=1 , k \neq i}}^{r} (x_i - x_k)}.
$$
\end{lemma}

\begin{example}
    Let $r=4$. The inverse $\bV^{-1}$ of the modified Vandermonde matrix $\bV$ in Lemma  \ref{L-2.2} is  given by
    $$\begin{pmatrix}
    	-\frac{x_2x_3x_4}{x_1(x_1-x_2)(x_1-x_3)(x_1-x_4)} & \frac{x_2x_3+x_2x_4+x_3x_4}{x_1(x_1-x_2)(x_1-x_3)(x_1-x_4)} & -\frac{x_2+x_3+x_4}{x_1(x_1-x_2)(x_1-x_3)(x_1-x_4)} & \frac{1}{x_1(x_1-x_2)(x_1-x_3)(x_1-x_4)} \\
    	-\frac{x_1x_3x_4}{x_2(x_2-x_1)(x_2-x_3)(x_2-x_4)} & \frac{x_1x_3+x_1x_4+x_3x_4}{x_2(x_2-x_1)(x_2-x_3)(x_2-x_4)} & -\frac{x_1+x_3+x_4}{x_2(x_2-x_1)(x_2-x_3)(x_2-x_4)} & \frac{1}{x_2(x_2-x_1)(x_2-x_3)(x_2-x_4)} \\
    	-\frac{x_1x_2x_4}{x_3(x_3-x_1)(x_3-x_2)(x_3-x_4)} & \frac{x_1x_2+x_1x_4+x_2x_4}{x_3(x_3-x_1)(x_3-x_2)(x_3-x_4)} & -\frac{x_1+x_2+x_4}{x_3(x_3-x_1)(x_3-x_2)(x_3-x_4)} & \frac{1}{x_3(x_3-x_1)(x_3-x_2)(x_3-x_4)} \\
    	-\frac{x_1x_2x_3}{x_4(x_4-x_1)(x_4-x_2)(x_4-x_3)} & \frac{x_1x_2+x_1x_3+x_2x_3}{x_4(x_4-x_1)(x_4-x_2)(x_4-x_3)} & -\frac{x_1+x_2+x_3}{x_4(x_4-x_1)(x_4-x_2)(x_4-x_3)} & \frac{1}{x_4(x_4-x_1)(x_4-x_2)(x_4-x_3)}
    \end{pmatrix}.\\
    $$
\end{example}

Lemma \ref{L-2.2} allows us to express the coefficients of a codeword of weight $\delta$ as rational functions of the locators. This leads to a necessary and sufficient condition for a BCH code $\bC_{(q,n,\delta,1)}$ to have minimum distance $d=\delta$.

\begin{theorem}\label{T-3.1}
	Let $m=\ord_n(q)$ and let $\beta$ be a primitive $n$-th root of unity in $\mathbb{F}_{q^m}$.
	The BCH code $\bC_{(q,n,\delta,1)}$ has true minimum distance $d=\delta$ if and only if there exist $\delta$ pairwise distinct elements
	$x_j=\beta^{i_j} \in \F_{q^m}^*$
	with $ 1 \le j \le \delta$ and $0 \le i_j \le n-1$
	such that
	\begin{equation}\label{e-0}
		S_j := \frac{\prod_{k=1, k \neq j}^{\delta-1} (1 - x_k)}{x_j \prod_{k=1, k \neq j}^{\delta-1} (x_j - x_k)} \in \F_q^*, \quad 1 \le j \le \delta-1.
	\end{equation}
	Furthermore, these elements $x_j$ are exactly the locators of a minimum-weight codeword.

\end{theorem}

\begin{proof}
\rm{
	We first prove the sufficiency.
	Suppose that such $\delta$ pairwise distinct elements exist.  
    Assume $i_\delta=0$ without loss of generality, which implies $x_\delta = \beta^{i_\delta} = 1$.
    To prove $d=\delta$, it suffices to construct a valid codeword $c(x)$ of weight $\delta$ with support $\{i_1, \dots, i_{\delta-1}, 0\}$.

    By the definition of BCH codes, a polynomial
    $c(x)= \sum_{j=1}^{\delta } c_{i_j} x^{i_j}$ with coefficients $c_{i_j} \in \F_q^*$ belongs to $\bC_{(q,n,\delta,1)}$ if and only if it satisfies the parity check equations,
    \begin{equation}\label{e-3.1-0}
    	c(\beta^s) = 0, \quad s=1,\dots,\delta-1.
    \end{equation} 
    Let $\{x_1, \dots, x_{\delta-1}, x_\delta \}$ be the locators of $c(x)$,  then \eqref{e-3.1-0} is equivalent to
    \begin{equation}\label{e-3.1-1}
    	\sum_{j=1}^{\delta} c_{i_j} (\beta^{s})^{i_j} = \sum_{j=1}^{\delta} c_{i_j} x_j^s = 0, \quad s=1, \dots, \delta-1.
    \end{equation}
   Then, \eqref{e-3.1-1} can be written as
    \begin{equation}\label{e-3.1-2}
    	\begin{pmatrix}
    		x_1 & x_2 & \cdots & x_\delta \\
    		x_1^2 & x_2^2 & \cdots & x_\delta^2 \\
    		\vdots & \vdots & \ddots & \vdots \\
    		x_1^{\delta-1} & x_2^{\delta-1} & \cdots & x_\delta^{\delta-1}
    	\end{pmatrix}
    	\begin{pmatrix}
    		c_{i_1} \\ c_{i_2} \\ \vdots \\ c_{i_\delta}
    	\end{pmatrix} = \mathbf{0}.
    \end{equation} 
    Let $\bV$ denote the $(\delta-1) \times (\delta-1)$ modified Vandermonde matrix of $x_1, \dots, x_{\delta-1}$ and let $(1, 1, \dots, 1)^\top = \mathbf{1}$. Then \eqref{e-3.1-2} becomes
    $$
    \bV \cdot \begin{pmatrix} c_{i_1} \\ \vdots \\ c_{i_{\delta-1}} \end{pmatrix} + \mathbf{1} \cdot c_{i_\delta} = \mathbf{0}.
    $$
    Given that $x_1, \dots, x_{\delta-1}$ are distinct and nonzero, $\bV$ is invertible. Thus,
    \begin{equation}\label{e-3.1-3}
    	\begin{pmatrix} c_{i_1} \\ \vdots \\ c_{i_{\delta-1}} \end{pmatrix} = - \bV^{-1} \cdot \mathbf{1} \cdot c_{i_\delta}.
    \end{equation} 
    By Lemma \ref{L-2.2}, the entries of $\bV^{-1} \cdot \bone$ are precisely $S_j$.
    Hence,
    \begin{equation}\label{e-3.1-4}
    	c_{i_j}=-c_{i_\delta} \cdot S_j, \quad 1 \le j \le \delta-1.
    \end{equation} 
    Since $S_j \in \F_q^*$ by assumption, choosing any $c_{i_\delta} \in \F_q^*$ ensures that all coefficients $c_{i_j}$ also belong to $\F_q^*$. Therefore $c(x)= \sum_{j=1}^{\delta } c_{i_j} x^{i_j}$ is a valid codeword in $\bC_{(q,n,\delta,1)}$. The minimum distance of $\bC_{(q,n,\delta,1)}$ is $\delta$.

    Conversely, suppose that $\bC_{(q,n,\delta,1)}$ has minimum distance $d=\delta$.
    This implies the existence of at least one codeword $c(x)$ with weight $\delta$.
    Let its support be $\{i_1, \dots, i_\delta\}$ and let $x_j = \beta^{i_j} \in \F_{q^m}$ for $1 \le j \le \delta$.
    Without loss of generality, we may shift the codeword so that $i_\delta = 0$ and thus $x_\delta = 1$.

    As $c(x)$ is a codeword, its coefficients $c_{i_j}$ must be non-zero elements in $\F_q$ and satisfy the parity-check system \eqref{e-3.1-3}.
    By \eqref{e-3.1-4},
    $$
    S_j=-c_{i_j} / c_{i_\delta}, \quad 1 \le j \le \delta-1,
    $$
    where $S_j$ is defined as \eqref{e-0}.
    Since $c_{i_j},c_{i_\delta} \in  \F_q^*$, it follows that $S_j \in \F_q^*$ for all $1 \le j \le \delta-1$.

    This completes the proof.
    \qed
}
\end{proof}

\begin{remark}
\rm{
	The algebraic approach in Theorem \ref{T-3.1} can be extended to weights $w>\delta$ through a Vandermonde system. However, in this case the resulting equations involve additional unknown coefficients, and one can no longer derive a closed-form condition depending solely on the locators. Hence, the characterization in terms of locators is restricted to the case $w=\delta$.
}
\end{remark}

Theorem \ref{T-3.1} shows that if there exists a set of $\delta$ distinct elements in $\F_{q^m}^*$ satisfying \eqref{e-0}, then the corresponding BCH code must contain a codeword of weight $\delta$. In particular, Condition \eqref{e-0} depends only on algebraic relations among these elements and remains valid under any field extension that contains them.  Specifically, we have the following useful result. 

\begin{theorem}\label{T-3.2}
    Let $\bC_{(q,n,\delta,1)}$ be a narrow-sense primitive BCH code of length $n = q^h - 1$ with minimum distance $d = \delta$. If $n' = q^m - 1$ where $m \equiv 0 \pmod{h}$, then the narrow-sense primitive BCH code $\bC'_{(q,n',\delta,1)}$ also has minimum distance $d' = \delta$.
\end{theorem}

\begin{proof}
\rm{
	Assume that $\bC_{(q,n,\delta,1)}$ has minimum distance  $d=\delta$. By Theorem \ref{T-3.1}, there exist $\delta-1$ distinct nonzero elements $x_{1}, \ldots, x_{\delta-1} \in \F_{q^h}^* \setminus \{1\}$ such that
	
	\begin{equation}\label{e-3.2-0}
		S_j:=\frac{\prod_{k=1, k \neq j}^{\delta-1} (1 - x_k)}{x_j\prod_{k=1, k \neq j}^{\delta-1} (x_j - x_k)} \in \F_{q}^*, \quad 1 \le j \le \delta-1.
	\end{equation}
	
	For any $m \equiv 0 \pmod{h}$, we have $\F_{q^h} \subseteq \F_{q^m}$.
	Thus, $x_{1}, \ldots, x_{\delta-1}$ also belong to $\F_{q^m}^* \setminus \{1\}$ and satisfy \eqref{e-3.2-0}. By Theorem \ref{T-3.1}, this implies the existence of a codeword of weight $\delta$ in $\bC'_{(q,n',\delta,1)}$. The minimum distance of $\bC'_{(q,n',\delta,1)}$ is $d'=\delta$.
	\qed
}
\end{proof}

Theorems \ref{T-3.1} and \ref{T-3.2} will be used later for determining the minimum distances of some families of BCH codes. 

\section{Primitive BCH codes with $d=\delta$}\label{s4}

In this section, we construct explicit infinite families of ternary, quaternary and $q$-ary primitive BCH codes with $d=\delta$.

\subsection{Ternary and quaternary BCH codes with $\delta \in \{5,6,7,8\}$}

\begin{theorem}\label{T-ternary}
	Let $m \ge 3$. The ternary narrow-sense primitive BCH code $\bC_{(3, 3^m-1, \delta, 1)}$ has the following parameters:
	\begin{itemize}
		\item For $\delta=5$, if $2 \mid m$ or $3 \mid m$, then the code has parameters $[3^m-1, 3^m-1-3m, 5]_3$. This code is almost distance-optimal.
		\item For $\delta=7$, if $3 \mid m$ or $4 \mid m$,  then the code has parameters $[3^m-1, 3^m-1-4m, 7]_3$.
		\item For $\delta=8$, if $3 \mid m$,  then the code has parameters $[3^m-1, 3^m-1-5m, 8]_3$.
	\end{itemize}
\end{theorem}

\begin{proof}
	\rm{
		The dimensions follow directly from Lemma \ref{L-dim}.
		Theorem \ref{T-3.2} shows that if $\mathbf{C}_{(3,3^h-1,\delta,1)}$ has minimum distance $d=\delta$, then $\mathbf{C}_{(3,3^m-1,\delta,1)}$ also  has $d=\delta$ for any $m$ divisible by $h$.
		SageMath and Magma computations have confirmed the minimum distance in the required base cases.
		For $\delta=5$, $\bC_{(3,8,5,1)}$ with $h=2$ and $\bC_{(3,26,5,1)}$ with $h=3$ have $d=5$.
		For $\delta=7$, $\bC_{(3,26,7,1)}$ with $h=3$ and $\bC_{(3,80,7,1)}$ with $h=4$ have $d=7$.
		For $\delta=8$, $\bC_{(3,26,8,1)}$ with $h=3$ has $d=8$.
		It remains to show that the code $[3^m-1, 3^m-1-3m, 5]_3$ is almost distance-optimal for $m \ge 3$.
		
		For $n=3^m-1$, the sphere packing bound gives the following restriction on the parameters of $\bC_{(3, 3^m-1, 5, 1)}$,
		\begin{equation}\label{e-4.1-0}
			\sum_{i=0}^{\left\lfloor \frac{d-1}{2} \right\rfloor} 2^i \binom{n}{i} \leq 3^{3m}.
		\end{equation}
		Suppose $m \ge 3$ and $d \ge 7$, then \eqref{e-4.1-0} leads to
%		\begin{equation}\label{e-4.1-1}
			$$1+2n+\frac{4n(n-1)}{2}+\frac{8n(n-1)(n-2)}{6} \le (n+1)^3.$$
%		\end{equation} 		
		After a simplification, we have
		\begin{equation}\label{e-4.1-2}
			n^2-15n-1 \le 0.
		\end{equation}
		When $n=3^m-1\ge26$, \eqref{e-4.1-2} does not hold. Therefore, there is no $[3^m-1, 3^m-1-3m, 7]_3$ code. Then any linear code with parameters $[3^m-1, 3^m-1-3m, 5]_3$ is almost distance-optimal with respect to the sphere packing bound.
		\qed
		
	}
\end{proof}

\begin{example}
	\rm{
		For $m=3$ and $\delta=5$, the code  $\bC_{(3,26,5,1)}$ has parameters $[26,17,5]_3$ and is almost distance-optimal according to \cite{codetable}.
		More codes in Theorem \ref{T-ternary} are listed in Table \ref{tab-4.1}.
		We compare the minimum distance with the best known value $d_{best}$ in \cite{codetable} for codes with the same parameters. The symbol ``/'' indicates that the best known distance is not available in \cite{codetable}. This notation applies to all subsequent tables.
		All numerical examples were produced with SageMath and Magma.
	}
\end{example}

\begin{table}[htbp]
	\centering
	\caption{\label{tab-4.1} Examples of the code in Theorem \ref{T-ternary}}
	\begin{tabular}{|l|l|l|l|}
		\hline
		$\delta$& $m$ & $\mathbf{C}_{(3,3^m-1,\delta,1)}$ & $d_{best}$ \\  \hline
		\multirow{3}{*}{5}
		& 3 & $[26,17,5]_3$     &  6\\ \cline{2-4}
		& 4 & $[80,68,5]_3$     &  6 \\ \cline{2-4}
		& 6 & $[728,710,5]_3$   &  / \\ \hline
		
		\multirow{3}{*}{7}
		& 3 & $[26,14,7]_3$   &  7 \\ \cline{2-4}
		& 4 & $[80,64,7]_3$   &  8 \\ \cline{2-4}
		& 6 & $[728,704,7]_3$   &  / \\ \hline
		
		\multirow{2}{*}{8}
		& 3 & $[26,11,8]_3$   &  9 \\ \cline{2-4}
		& 6 & $[728,698,8]_3$   &  / \\ \hline
	\end{tabular}
\end{table}

\begin{theorem}\label{T-quaternary}
	Let $m \ge 2$ be an integer such that $2 \mid m$ or $3 \mid m$. The quaternary narrow-sense primitive BCH code $\bC_{(4, 4^m-1, \delta, 1)}$ has the following parameters:
	\begin{itemize}
		\item For $\delta=5$, the code has parameters $[4^m-1, 4^m-1-3m, 5]_4$. This code is almost distance-optimal.
		\item For $\delta=6$, the code has parameters $[4^m-1, 4^m-1-4m, 6]_4$. This code is near distance-optimal.
		\item For $\delta=7$, the code has parameters $[4^m-1, 4^m-1-5m, 7]_4$.
	\end{itemize}
\end{theorem}

\begin{proof}
	\rm{
		The dimensions follow directly from Lemma \ref{L-dim}.
		Theorem \ref{T-3.2} shows that if $\mathbf{C}_{(4,4^h-1,\delta,1)}$ has minimum distance $d=\delta$, then $\mathbf{C}_{(4,4^m-1,\delta,1)}$ also  has $d=\delta$ for any $m$ divisible by $h$.
		SageMath and Magma computations have confirmed the minimum distance in the required base cases.
		For $\delta=5$, $\bC_{(4,15,5,1)}$ with $h=2$ and $\bC_{(4,63,5,1)}$ with $h=3$ have $d=5$.
		For $\delta=6$, $\bC_{(4,15,6,1)}$ with $h=2$ and $\bC_{(4,63,6,1)}$ with $h=3$ have $d=6$.
		For $\delta=7$, $\bC_{(4,15,7,1)}$ with $h=2$ and $\bC_{(4,63,7,1)}$ with $h=3$ have $d=7$.
		It remains to show the almost distance optimality of the code.
		
		For $n=4^m-1$, the sphere packing bound gives the following restriction on the parameters of $\bC_{(4, 4^m-1, 5, 1)}$,
		\begin{equation}\label{e-4.2-0}
			\sum_{i=0}^{\left\lfloor \frac{d-1}{2} \right\rfloor} 3^i \binom{n}{i} \leq 4^{3m}.
		\end{equation}		
		Suppose $m \ge 2$ and $d \ge 7$, then \eqref{e-4.2-0} leads to
		$$1+3n+\frac{9n(n-1)}{2}+\frac{27n(n-1)(n-2)}{6} \le (n+1)^3.$$ 		
		After a simplification, we have
		\begin{equation}\label{e-4.2-2}
			7n^2-24n+9 \le 0.
		\end{equation}
		
		When $n=4^m-1\ge15$, \eqref{e-4.2-2} does not hold. Therefore, there is no $[4^m-1, 4^m-1-3m, 7]_4$ code, then any  $[4^m-1, 4^m-1-3m, 5]_4$ linear code is almost distance-optimal.
		Similarly, no $[4^m-1,4^m-1-4m,9]_4$ codes exist, then any $[4^m-1,4^m-1-4m,6]_4$ linear code is  near distance-optimal.
		\qed
		
	}
\end{proof}

\begin{example}
	\rm{
		For $m=3$ and $\delta=6$, the code  $\bC_{(4,63,6,1)}$ has parameters $[63,51,6]_4$, which are the best known parameters according to \cite{codetable}. More codes in Theorem \ref{T-quaternary} are listed in Table \ref{tab-4.2}.  All numerical examples were computed with both SageMath and Magma.
	}
\end{example}

\begin{longtable}{|l|l|l|}
	\caption{\label{tab-4.2} Examples of the code in Theorem \ref{T-quaternary} } \\ \hline
	$\delta$& $\bC_{(4,4^m-1,\delta,1)}$ & $d_{best}$ \\  \hline
	\multirow{3}{*}{5} & $[15,9,5]_4$     &  5\\ \cline{2-3}
	& $[63,54,5]_4$     &  5\\ \cline{2-3}
	& $[255,243,5]_4$   &  5 \\ \hline
	
	\multirow{3}{*}{6} & $[15,8,6]_4$   &  6 \\ \cline{2-3}
	& $[63,51,6]_4$   &  6 \\ \cline{2-3}
	& $[255,239,6]_4$   &  6 \\ \hline
	
	\multirow{3}{*}{7} & $[15,6,7]_4$   &  8 \\ \cline{2-3}
	& $[63,48,7]_4$   &  8 \\ \cline{2-3}
	& $[255,235,7]_4$   &  8 \\ \hline
	
\end{longtable}

\subsection{A family of $q$-ary BCH codes with variable $\delta$}

In this subsection, we construct $q$-ary BCH codes with $d=\delta$ for variable $\delta$. For any $2 \le \delta \le q-1$, we have the following result.

\begin{theorem}\label{T-delta<q}
    Let $q$ be a prime power and let $m \ge 1$ be a positive integer. For any $2 \le \delta \le q-1$, the narrow-sense BCH code  $\bC_{(q,q^m-1,\delta,1)}$ has parameters
    $$[q^m-1, q^m-1-m (\delta-1), \delta]_q.$$
\end{theorem}

\begin{proof}
\rm{
    The dimension follows directly from Lemma \ref{L-dim}. By the BCH bound, the minimum distance $d$ is at least $\delta$.

    We first prove that the conclusion on the minimum distance holds for $m=1$.
    By Theorem~\ref{T-3.1}, it suffices to prove that there exist distinct nonzero elements $x_{1}, \ldots, x_{\delta-1} \in \F_{q}^* \setminus \{1\}$ such that

    \begin{equation}\label{e-4.7-1}
        S_j:=\frac{\prod_{k=1, k \neq j}^{\delta-1} (1 - x_k)}{x_j\prod_{k=1, k \neq j}^{\delta-1} (x_j - x_k)} \in \F_{q}^*, \quad 1 \le j \le \delta-1.
    \end{equation}
    Choose $\delta-1$ distinct $x_j\in \F_q^* \setminus \{1\}.$ Then the terms $x_j$, $1-x_j$, and $x_j-x_k$ are all in $\F_q^*$, so $S_j \in \F_q^*$ for $1 \le j \le \delta-1$. Thus the conclusion on the minimum distance is true for $m=1$.  
    By Theorem \ref{T-3.2}, any code $\bC_{(q,q^m-1,\delta,1)}$ with $m>1$ has minimum distance $d=\delta$.
    This completes the proof.
    \qed

}

\end{proof}

\begin{remark}
\rm{
	For $m=1$, $\mathbf{C}_{(q, q-1, \delta, 1)}$ is a Reed-Solomon code and its minimum distance is directly given by $\delta$. The explicit construction of $x_j$ is provided to satisfy $S_j \in \mathbb{F}_q^*$, serving as the necessary base case for the extension to $m>1$ via Theorem \ref{T-3.2}.
}
\end{remark}

\begin{example}
{\rm 
Some examples of the code $\bC_{(q,q^m-1,\delta,1)}$ for $2 \le \delta \le q-1$ are listed below.
\begin{itemize}
\item When $(q, m, \delta)=(3,2,2)$, the code has parameters $[8,6,2]_3$ and is distance-optimal 
according to the codetable \cite{codetable}.
\item When $(q, m, \delta)=(3,3,2)$, the code has parameters $[26,23,2]_3$ and is distance-optimal 
according to the codetable \cite{codetable}.  
\item When $(q, m, \delta)=(4,2,3)$, the code has parameters $[15,11,3]_4$, while the distance-optimal 
linear code in the codetable \cite{codetable} has parameters $[15,11,4]_4$ and is not known to be cyclic.
\item When $(q, m, \delta)=(4,3,3)$, the code has parameters $[63, 57, 3]_4$, while the distance-optimal 
linear code in the codetable \cite{codetable} has parameters $[63,57,4]_4$ and is not known to be cyclic.
\end{itemize}

}
\end{example}

\subsection{The $q$-ary BCH codes $\bC_{(q,q^m-1,\delta,1)}$ with $\delta=q^t+1$}

In this subsection, we show that the $q$-ary BCH code $\bC_{(q,q^m-1,\delta,1)}$ has minimum distance $d=\delta=q^t+1$ under certain condition.  
A polynomial over $\F_q$ is said to be separable if its roots in the algebraic closure of $\F_q$ are distinct. The following lemma will be used in the sequel.

\begin{lemma}\label{L-p^t+1}
	Let $q$ be a power of a prime $p$. For any integer $t \ge 1$, the polynomial $$L(x) = x^{q^t} - x^{q^t-1} + 1 \in \mathbb{F}_q[x]$$
	is separable, and its splitting field is $\F_{q^{pt}}$.
	
\end{lemma}

\begin{proof}
\rm{
	Let $Q = q^t$. Then $L(x) = x^Q - x^{Q-1} + 1$. Since $Q \equiv 0 \pmod p$,
	$$\begin{aligned}
		L'(x) &= Q x^{Q-1} - (Q-1)x^{Q-2} = x^{Q-2}.
	\end{aligned}$$
	The only root of  $L'(x)$ is $x = 0$, and $L(0) = 1 \neq 0$.
	Since $L(x)$ and $L'(x)$ share no common roots, $\gcd(L(x), L'(x)) = 1$. Therefore, $L(x)$ is separable and has no multiple roots.
	
	The reciprocal polynomial of $L(x)$ is
	$$L^*(x) = x^Q L(1/x) = x^Q (x^{-Q} - x^{-(Q-1)} + 1) = 1 - x + x^Q.$$
	The splitting field of $L(x)$ is identical to the splitting field of $L^*(x)$. Let $\alpha$ be a root of $L^*(x)$. Then,
	$$\alpha^Q = \alpha - 1.$$
	
	By induction on an integer $k \ge 1$, we obtain
	$$\begin{aligned}
		\alpha^{Q^k} &= \alpha - k.
	\end{aligned}$$
	
	Thus, $\alpha \in \F_{Q^k}$ holds  if and only if
	$$\alpha^{Q^k} = \alpha \iff \alpha - k = \alpha \iff k \equiv 0 \pmod p.$$
	
	The smallest positive integer $k$ satisfying $k \equiv 0 \pmod{p}$ is $p$. Consequently, the splitting field of $L(x)$ is $\F_{Q^p}=\F_{q^{pt}}$.
	\qed

}
\end{proof}

\begin{theorem}\label{T-q^t+1}
    Let $q$ be a power of a prime $p$ and $t$ a positive integer.  For any $m \equiv 0 \pmod{pt}$, the narrow-sense primitive BCH code $\bC_{(q,q^{m}-1,q^t+1,1)}$ has minimum distance $d=q^t+1$ and dimension $q^{m}-1-m q^{t-1}(q-1)$.

\end{theorem}

\begin{proof}
\rm{

	The dimension follows directly from Lemma \ref{L-dim}. By the BCH bound, the minimum distance $d \ge q^t+1$.
	We first prove the conclusion for the case $m=pt$. By Theorem~\ref{T-3.1}, it suffices to prove that there exist distinct nonzero elements $x_{1}, \ldots, x_{q^t} \in \F_{q^m}^* \setminus \{1\}$ such that
	
	\begin{equation}\label{e-4.9-1}
		S_j:=\frac{\prod_{k=1, k \neq j}^{q^t} (1 - x_k)}{x_j\prod_{k=1, k \neq j}^{q^t} (x_j - x_k)} \in \F_{q}^*, \quad 1 \le j \le q^t.
	\end{equation}
	
    Let $L(x) = x^{q^t} - x^{q^t-1} + 1 \in \F_q[x]$.
    By Lemma \ref{L-p^t+1}, $L(x)$ is separable and its splitting field is $\F_{q^{pt}}$.
    Let $x_j=\beta^{i_j}$ for $1 \le j \le q^t$ be the roots of $L(x)$ over $\F_{q^{pt}}$, where $\beta$ is the primitive element of $\F_{q^{pt}}$.
    Since $0$ and $1$ are not the roots of $L(x)$, all roots of $L(x)$ lie in $\F_{q^{pt}}^* \setminus \{1\}$.

    Since $L'(x)=x^{q^t-2}$, we obtain
    \begin{equation*}\label{e-4.9-2}
    	x(1 - x)L'(x) = x^{q^{t}-1}-x^{q^{t}}=1-L(x).
    \end{equation*} 
    Hence, for every root $x_j$, we have
    \begin{equation}\label{e-4.9-3}
    	x_j(1 - x_j)L'(x_j) = 1.
    \end{equation}
    Using the factorization $L(x) = \prod_{k=1}^{q^t} (x - x_k)$, we obtain
	\begin{equation*}\label{e-4.8-2}
		\prod_{{k=1 , k \neq j}}^{q^t} (1 - x_k) = \frac{L(1)}{1 - x_j}, \quad
	 \quad \prod_{{k=1 , k \neq j}}^{q^t} (x_j - x_k) =L'(x_j).
	\end{equation*}
Consequently,  
    \begin{equation*}\label{e-4.9-4}
    	S_j = \frac{\frac{L(1)}{1-x_j}}{x_j L'(x_j)} = \frac{L(1)}{x_j(1 - x_j)L'(x_j)}.
    \end{equation*}

    From $L(1) = 1$ and \eqref{e-4.9-3}, it follows that $S_j = 1\in \F_q^*$ for all $1 \le j \le q^t$. According to Theorem \ref{T-3.1}, the coefficients satisfy $c_{i_j}=-c_{i_\delta}\cdot S_j$. Setting $c_{i_\delta}=-1$ yields $c_{i_j}=1$ for all $1 \le j \le q^t$. Hence,
    $$c(x)=\sum_{j=1}^{q^t} x^{i_j}-1$$ is a codeword of weight $q^t+1$. The minimum distance of $\bC_{(q,q^{pt}-1,q^t+1,1)}$ is $q^t+1$.
    By Theorem \ref{T-3.2}, any code $\bC_{(q,q^{m}-1,\delta,1)}$ with $m \equiv 0 \pmod{pt}$ also has minimum distance $d=\delta$.
    This completes the proof.
    \qed
}
\end{proof}
\begin{remark}
	\rm{
		In \cite{Ding0} the authors conjectured that the BCH code $\bC_{(q,q^m-1,q^{t}+1,1)}$ always has $d=d_{B}$, where $d_B$ denotes the Bose distance of the code, see Conjecture \ref{c-0}.
%		For $t \le m/2$, it was proved in \cite{Ding0} that $d_B=\delta=q^t+1$.
		According to the classical result in \cite[p.~247]{Betten}, $d=\delta$ holds for a narrow-sense BCH code whenever $\delta \mid n$.
		Then Conjecture \ref{c-0} holds when $m \equiv 0 \pmod{2t}$ or $t \equiv 0 \pmod{m-t}$, as both conditions ensure $\delta \mid n$.
		However, the cases where $\delta \nmid n$ remained open.
		Theorem~\ref{T-q^t+1} confirms Conjecture \ref{c-0} for all $m \equiv 0 \pmod{pt}$ where $p$ is the characteristic of $\F_q$. This is achieved by explicitly identifying $H(x)=(x-1)L(x)$ as the locator polynomial of a   codeword with weight $q^t+1$ in $\bC_{(q,q^m-1,q^{t}+1,1)}$.
		It is worth highlighting that $q^t+1 \mid q^m-1$  when $m/pt$ is even but $q^t+1 \nmid q^m-1$ when $m/pt$ is odd.
		Therefore, the codes covered in Theorem \ref{T-q^t+1} include the nontrivial cases where $\delta \nmid n$.

	}
\end{remark}

\begin{example}
\rm{
Examples of the code in Theorem \ref{T-q^t+1} are given in Table \ref{tab-q^t+1}. All numerical examples were 
produced using SageMath.

\begin{itemize}
    \item Let $(q,t,m)=(3,1,3)$, $\delta=4$. Let $\beta$ be a primitive element of $\F_{3^3}$. The roots of $L(x)=x^3-x^2+1$ in its splitting field $\F_{3^3}$ are $\{\beta^{25},\beta^{23},\beta^{17}\}$. Thus, the corresponding nonzero positions are $\{i_1,i_2,i_3\}=\{25,23,17\}$. The polynomial $c(x)=x^{25}+x^{23}+x^{17}+2$ is a codeword of weight $4$ in $\bC_{(3,26,4,1)}$. The code $\bC_{(3,26,4,1)}$ has parameters $[26,20,4]_3$. This code is distance-optimal according to the codetable \cite{codetable}.

    \item
    Let $(q,t,m)=(4,2,4)$, $\delta=17$. Let $\beta$ be a primitive element of $\F_{4^4}$. The splitting field of the polynomial $L(x)=x^{16}+x^{15}+1$ is $\F_{4^4}$. Similarly, we construct  a polynomial $c(x)=x^{248} + x^{241} + x^{227} + x^{199} + x^{144} + x^{143} + x^{132} + x^{124} + x^{72} + x^{66} + x^{62} + x^{36} + x^{33} + x^{31} + x^{18} + x^{9} + 1$, which is a codeword of weight $17$ in $\bC_{(4,255,17,1)}$. The code $\bC_{(4,255,17,1)}$ has parameters $[255,207,17]_4$ and is a best known code according to the codetable \cite{codetable}.

    \item
    Let  $(q,t,m)=(5,1,5)$, $\delta=6$. Let $\beta$ be a primitive element of $\F_{5^5}$. The roots of $L(x)=x^5-x^4+1$ in its splitting field $\F_{5^5}$ are $\{\beta^{2342},\beta^{2338},\beta^{2318},\beta^{2218},\beta^{1718}\}$. Thus, the corresponding nonzero positions are  $\{i_1,i_2,i_3,i_4,i_5\}=\{2342, 2338, 2318, 2218, 1718\}$. The polynomial $c(x)=x^{2342} + x^{2338} + x^{2318} + x^{2218} + x^{1718} + 4$ is a codeword of weight $6$ in $\bC_{(5,5^5-1,6,1)}$. The code $\bC_{(5,5^5-1,6,1)}$ has parameters $[3124, 3104, 6]_5$.

\end{itemize}

}
\end{example}

\begin{longtable}{|l|l|l|l|l|}
	\caption{\label{tab-q^t+1} Examples of the code in Theorem \ref{T-q^t+1}}\\
	\hline
	$q$ & $t$ & $m$ & $\bC_{(q,q^{m}-1,q^t+1,1)}$ & $d_{best}$ \\
	\hline
	\endfirsthead
	
	\hline
	$q$ & $t$ & $m$ & $\bC_{(q,q^{m}-1,q^t+1,1)}$ & $d_{best}$ \\
	\hline
	\endhead
	
	\hline
	\endfoot
	
	\multirow{7}{*}{2}
	& \multirow{3}{*}{1} & 4 & $[15,11,3]_2$  &  3 \\ \cline{3-5}
	& & 6 & $[63,57,3]_2$  &  3 \\ \cline{3-5}
	& & 8 & $[255,247,3]_2$  &  3 \\ \cline{2-5}
	& \multirow{2}{*}{2} & 4 & $[15,7,5]_2$  &  5 \\ \cline{3-5}
	& & 8 & $[255,239,5]_2$  &  5 \\ \cline{2-5}
	
	& 3 & 6 & $[63,39,9]_2$  &  9\\ \cline{2-5}
	& 4 & 8 & $[255,191,17]_2$  &  17\\ \hline

	\multirow{3}{*}{3}
	& \multirow{2}{*}{1} & 3 & $[26,20,4]_3$   &  4 \\ \cline{3-5}
	&  & 6 & $[728,716,4]_3$  &  / \\ \cline{2-5}
	
	& 2 & 6 & $[728,692,10]_3$  &  / \\ \hline
		
	\multirow{3}{*}{4}
	& \multirow{2}{*}{1} & 2 & $[15,9,5]_4$   &  5 \\ \cline{3-5}
	&  & 4 & $[255,243,5]_4$  &  5 \\ \cline{2-5}
	
	& 2 & 4 & $[255,207,17]_4$  &  17 \\ \hline
	
	5 & 1 & 5 & $[3124,3104,6]_5$  &  / \\ \hline
	
\end{longtable}

\section{A family of non-primitive BCH codes with $d=\delta$}\label{s5} 

In this section, we extend the method of constructing the locator polynomial developed in Theorem \ref{T-q^t+1} to the non-primitive cases.
Let $p$ be prime and $q=p^e$ with $p \nmid e$ and let $\lambda$ be a divisor of $q-1$.
We construct an infinite family of non-primitive BCH codes $\bC_{(q,\frac{q^p-1}{\lambda},p+1,1)}$ with $d=\delta=p+1$.

The main idea is to use the polynomial $Q(x)=x^p+x^{p-1}+ \dots + x -1\in\F_q[x]$. We will show that $H(x)=(x-1)Q(x)$ corresponds to the locator polynomial of a codeword of weight $p+1$ in $\bC_{(q,\frac{q^{p}-1}{\lambda},p+1,1)}$. We begin with several auxiliary lemmas.

\begin{lemma}\label{L-4-1}
	Let $p$ be a prime and let $q=p^e$ with $p \nmid e$. Then the polynomial $A(x)=x^p-x-1\in\F_q[x]$ is irreducible and separable over $\F_q$.
\end{lemma}

\begin{proof}
	\rm{
		Since $A'(x) = p x^{p-1} - 1=-1$, we have  $\gcd(A(x), A'(x)) = \gcd(A(x), -1) = 1$, and thus $A(x)$ is separable over $\F_q$.
		
		To prove the irreducibility, assume for contradiction that $A(\alpha)=0$ for some  $\alpha \in \F_q$. Then $\alpha^p-\alpha=1$. Recall that the trace map $\text{Tr}_{\mathbb{F}_q/\mathbb{F}_p}$ is defined as $\text{Tr}_{\mathbb{F}_q/\mathbb{F}_p}(x) = x + x^p + \dots + x^{p^{e-1}}$ for any $x \in \mathbb{F}_q$. Apply the trace map $\Tr_{\F_q/\F_p}$ to both sides. As $\Tr_{\F_q/\F_p}(\alpha^p)=\Tr_{\F_q/\F_p}(\alpha)$ for every $\alpha\in\F_q$, we have
		\begin{equation}\label{e-5.1-L-0}
			\Tr_{\F_q/\F_p}(\alpha^p-\alpha)=0.
		\end{equation}
		
		However, since $p \nmid e$,
		$$
		\Tr_{\F_q/\F_p}(1)=e \not\equiv 0 \pmod{p},
		$$
		a contradiction with \eqref{e-5.1-L-0}.
		Thus, $A(x)$ has no root in $\F_q$. Since $\deg (A(x))=p$ is prime, $A(x)$ is  irreducible.
		\qed
	}
\end{proof}

\begin{lemma}\label{L-4-2}
	Let $p$ be a prime and $q=p^e$ with $p \nmid e$. Then the polynomial $Q(x)=x^p+x^{p-1}+ \dots + x -1\in\F_q[x]$ is irreducible and separable over $\F_q$.
\end{lemma}

\begin{proof}
	\rm{
		Consider the translation $x\mapsto x+1$. Then,
		$$Q(x+1) = \frac{(x+1)^{p+1}-1}{(x+1)-1} - 2.$$ 		
		In characteristic $p$, $(x+1)^{p+1}=x^{p+1}+x^p+x+1$. Hence,
		$$Q(x+1) = \frac{x^{p+1}+x^p + x}{x} -2 = x^p + x^{p-1} - 1.$$ 		
		Let $R(x):=Q(x+1)=x^p + x^{p-1} - 1$. The reciprocal polynomial of $R(x)$ is
		$$
		R^*(x)=x^p R(\frac{1}{x})=x^p \left[ \left(\frac{1}{x}\right)^p + \left(\frac{1}{x}\right)^{p-1} - 1 \right]=-(x^p-x-1)=-A(x).
		$$
		
		Since irreducibility is preserved under reciprocality and $A(x)$ is irreducible by Lemma \ref{L-4-1}, $R(x)$ is irreducible.		
		Moreover, $R'(x)=-x^{p-2}\neq 0$, so $\gcd(R(x),R'(x))=1$ and $R(x)$ is separable.
		Note that translation preserves separability and irreducibility. Hence $Q(x)$ is  irreducible and separable over $\F_q$.
		\qed
		
	}
\end{proof}

\begin{theorem}\label{T-p+1/lambda}
    Let $p$ be an odd prime, $q=p^e$ with $p \nmid e$ and $\lambda \mid q-1$. The narrow-sense BCH code $\bC_{(q,\frac{q^p-1}{\lambda},p+1,1)}$ has minimum distance $d=p+1$ and dimension
    $$
    k=\begin{cases}
    	\frac{q^p-1}{\lambda }-p(p-1),  & \text{ if } e= 1,  \\
    	\frac{q^p-1}{\lambda }-p^2,  & \text{ if } e \ge 2.
    \end{cases}
    $$
    %    [\frac{q^p-1}{\lambda},\frac{q^p-1}{\lambda}-1-p(p-1),p+1]_q.
\end{theorem}

\begin{proof}
\rm{
	The dimension follows directly from Lemma \ref{L-dim}. By the BCH bound, the minimum distance $d$ is at least $p+1$. To show $d=p+1$, it suffices to prove that there exist distinct nonzero elements $x_j=\gamma^{i_j}$ for $1 \le j \le p$, where $\gamma$ is an $n$-th primitive root in $\F_{q^p}$, such that
	
	\begin{equation}\label{e-5.1-0}
		S_j:=\frac{\prod_{k=1, k \neq j}^{p} (1 - x_k)}{x_j\prod_{k=1, k \neq j}^{p} (x_j - x_k)} \in \F_{q}^*, \quad 1 \le j \le p.
	\end{equation}
	
	Let $Q(x)=x^p+x^{p-1}+ \dots + x -1\in\F_q[x]$.  By Lemma \ref{L-4-2}, $Q(x)$ is irreducible and separable over $\F_q$, then its splitting field is $\F_{q^{p}}$.
	Let $x_j=\beta^{i_j}$ for $1 \le j \le p$ be the roots of $Q(x)$ in $\F_{q^{p}}$, where $\beta$ is a primitive element of $\F_{q^{p}}$.
	
	Since $Q(x)$ is a monic irreducible polynomial of degree $p$ over $\F_q$, the norm of any root $x_j$ relative to the extension $\F_{q^p}/\F_q$ is equal to the product of all roots.
	Recall that for any $y \in \mathbb{F}_{q^p}$, the norm $\N_{\mathbb{F}_{q^p}/\mathbb{F}_q}(y)$ is defined as $y^{1+q+\dots+q^{p-1}}$.
%	Recall that for any $y \in \mathbb{F}_{q^p}$, the norm is defined as $y^{1+q+\dots+q^{p-1}}$.
    Thus, by Vieta's formula,
    $$
    \N_{\F_{q^p / \F_q}}(x_j)= x_j^{\frac{q^p-1}{q-1}} = \prod_{k=1}^p x_k=(-1)^p\cdot Q(0).
    $$
    Since $p$ is an odd prime and $Q(0)=-1$, $\N_{\F_{q^p / \F_q}}(x_j)=1$, we have
    $$x_j^{\frac{q^p-1}{q-1}}=\beta^{i_j \cdot \frac{q^p-1}{q-1}}=1.$$
    Since $\beta$ has multiplicative order $q^p-1$, it follows that $(q-1) \mid i_j$ for each $i_j$, and thus $\lambda \mid i_j$ for any $\lambda \mid q-1$. \\

    Let $n=\frac{q^p-1}{\lambda}$, and let $\gamma=\beta^\lambda$, which is a primitive $n$-th root of unity. Since $\lambda \mid i_j$, it follows that $\{x_1,\dots,x_p\}=\{\beta^{i_1}, \dots, \beta^{i_p}\}=\{\gamma^{i_1/\lambda}, \dots, \gamma^{i_p/\lambda}\}$. 	
	A direct computation shows that
	\begin{equation}\label{e-5.1-2}
		x(1 - x)Q'(x) =1+Q(x).
	\end{equation} 	
	Evaluating \eqref{e-5.1-2} at any root $x_j$ gives
	\begin{equation}\label{e-5.1-3}
		x_j(1 - x_j)Q'(x_j) = 1.
	\end{equation} 
	Using the factorization $Q(x) = \prod_{k=1}^{p} (x - x_k)$, we obtain
	\begin{equation*}\label{e-4.8-2-10}
		\prod_{{k=1 , k \neq j}}^{p} (1 - x_k) = \frac{Q(1)}{1 - x_j}, \quad \quad 
		\prod_{{k=1 , k \neq j}}^{p} (x_j - x_k) =Q'(x_j).
	\end{equation*}
Hence, we have 
	\begin{equation*}\label{e-5.1-4}
		S_j = \frac{\frac{Q(1)}{1-x_j}}{x_j Q'(x_j)} = \frac{Q(1)}{x_j(1 - x_j)Q'(x_j)}.
	\end{equation*} 
	From $Q(1) = -1$ and \eqref{e-5.1-3}, it follows that $S_j = -1\in \F_q^*$ for all $1 \le j \le p$.
	According to Theorem \ref{T-3.1}, the coefficients satisfy $c_{i_j}=-c_{i_\delta}\cdot S_j$. Setting $c_{i_\delta}=1$ yields $c_{i_j}=1$ for all $1 \le j \le p$.
    Hence,
	$$c(x)=\sum_{j=1}^p x^{i_j/\lambda}+1$$ is a codeword of weight $p+1$ in $\bC_{(q,\frac{q^p-1}{\lambda},p+1,1)}$. The minimum distance of $\bC_{(q,\frac{q^p-1}{\lambda},p+1,1)}$ is $p+1$.
    This completes the proof.
    \qed
}
\end{proof}

\begin{example}
\rm{
Examples of the code in Theorem \ref{T-p+1/lambda} are given in Table \ref{tab-p+1/lambda}. All numerical examples were computed using SageMath.

\begin{itemize}
    \item Let $(p,q,\lambda)=(3,3,2)$ and let $\beta$ be a primitive element of $\F_{3^3}$. The roots of $Q(x)=x^3+x^2+x-1$ in its splitting field $\F_{3^3}$ are $\{\beta^2,\beta^6,\beta^{18}\}$. Thus, the corresponding nonzero positions are $\{i_1,i_2,i_3\}=\{1,3,9\}$. The polynomial $c(x)=x^{9}+x^3+x^1+1$ is a codeword of weight $4$ in $\bC_{(3,13,4,1)}$. The code $\bC_{(3,13,4,1)}$ has parameters $[13,7,4]_3$. The corresponding distance-optimal code has parameters $[13,7,5]_3$ according to the codetable \cite{codetable}.

    \item Let $(p,q,\lambda)=(5,5,4)$ and let $\beta$ be a primitive element of $\F_{5^5}$. The roots of $Q(x)=x^5+x^4+x^3+x^2+x-1$ in its splitting field $\F_{5^5}$ are $\{\beta^4,\beta^{20},\beta^{100},\beta^{500},\beta^{2500}\}$. Thus, the corresponding nonzero positions are $\{i_1,i_2,i_3,i_4,i_5\}=\{1,5,25,125,625\}$. The polynomial $c(x)=x^{625} + x^{125} + x^{25} + x^5 + x + 1$ is a codeword of weight $6$ in $\bC_{(5,781,6,1)}$. The code $\bC_{(5,781,6,1)}$ has parameters $[781,761,6]_5$.

    \item
    Let $(p,q,\lambda)=(3,9,8)$ and let $\beta$ be a primitive element of $\F_{9^3}$. The roots of $Q(x)=x^3+x^2+x-1$ in its splitting field $\F_{9^3}$ are $\{\beta^{56},\beta^{168},\beta^{504}\}$. Thus, the corresponding nonzero positions are  $\{i_1,i_2,i_3\}=\{7,21,63\}$. The polynomial $c(x)=x^{63}+x^{21}+x^{7}+1$ is a codeword of weight $4$ in $\bC_{(9,91,4,1)}$. The code $\bC_{(9,91,4,1)}$ has parameters $[91, 82, 4]_9$. The corresponding best known code has parameters $[91, 82, 5]_9$ according to the codetable \cite{codetable}.
\end{itemize}

}
\end{example}

\begin{longtable}{|l|l|l|l|}
    \caption{\label{tab-p+1/lambda} Examples of the code in Theorem \ref{T-p+1/lambda}}\\ \hline
    $q$& $\lambda$ & $\bC_{(q,(q^p-1)/\lambda,p+1,1)}$ & $d_{best}$ \\  \hline
    % 3 & $[26,20,4]_3$     &  4 \\

    3 & 2 & $[13,7,4]_3$     &  5 \\ \hline
    % 5 & $[3124,3104,6]_5$   &  / \\
    \multirow{2}{*}{5} & 4 & $[781,761,6]_5$   &  / \\ \cline{2-4}
    & 2 & $[1562,1542,6]_5$   &  / \\ \hline

    % 9 & $[728,719,4]_8$   &  - \\
    \multirow{3}{*}{9} & 8 & $[91,82,4]_9$   &  5 \\ \cline{2-4}
    & 4 & $[182,173,4]_9$   &  / \\ \cline{2-4}
    & 2 & $[364,355,4]_9$   &  / \\ \hline
\end{longtable}

\section{Summary and concluding remarks}\label{s6} 

In this paper, we settled the minimum distances of a number of infinite families of narrow-sense primitive 
BCH codes and an infinite family of narrow-sense non-primitive BCH codes. In addition,   we confirmed a conjecture in \cite{Ding0} about the minimum distance of 
$\bC_{(q, q^m-1, q^t+1, 1)}$ under the condition $m \equiv 0 \pmod{pt}$, where $p$ denotes the characteristic of $\F_q$.   A summary of the main contributions of this paper was given in Section \ref{sec-contribu}.  

The minimum distances of these families of BCH codes were determined by using Theorems \ref{T-3.1} and \ref{T-3.2} via the locator polynomial method.   It will be good if this method can be generalized to other classes of BCH codes and be used to resolve the remaining cases of Conjecture \ref{c-0}. 

Notice that the locator polynomial method has been used as a tool for finding the true minimum distance of 
cyclic codes for many years (see, e.g., \cite{ACS92,MScode}).  However, it is not easy to use this method. 
Even if a BCH code has a conjectured minimum distance $d$,  it could be very difficult to confirm the conjectured minimum distance with the locator polynomial method due to the difficulty in finding a codeword of weight $d$ in the BCH code.  This explains why Conjecture \ref{c-0} is still open in certain cases.

\section*{Acknowledgements}
The authors are very grateful to the Associate Editor, Prof.  Ferruh \"Ozbudak,  and the referees for their comments that improved the presentation of this paper.

\end{document}